%% file: main.tex
\documentclass[11pt,letterpaper]{article}

\usepackage[margin=1in]{geometry}
\usepackage{pgf}
\usepackage{contour}
\usepackage[format=hang]{caption}
\contourlength{0.85pt}

\input{commands}

\usepackage{natbib}

\let\cite\citep
\usepackage{amsmath, amsthm}
\usepackage{hyperref}
\usepackage[capitalize]{cleveref}
\usepackage{autonum}
\usepackage{thmtools}

\makeatletter
\IfFormatAtLeastTF{2024-11-01}{
  \renewcommand*\@addtoreset[2]{%
    \bgroup
      \edef\aliasctr@@truelist{\aliasctr@follow{#2}}%
      \let\@elt\relax
      \expandafter\@cons\aliasctr@@truelist{{#1}}%
    \egroup
    \expandafter\xdef\csname theH#1\endcsname{%
      \expandafter\noexpand\csname theH#2\endcsname.%
      \noexpand\the\noexpand\value{#1}}%
  }
}{}
\makeatother

\AddToHook{cmd/appendix/before}{%
    \crefalias{section}{appendix}%
    \crefalias{subsection}{appendix}
}

\usepackage{thmtools}

\newcounter{thmctr}
\numberwithin{thmctr}{section}

\newtheorem{theorem}[thmctr]{Theorem}

\newtheorem{lemma}[thmctr]{Lemma}
\newtheorem{proposition}[thmctr]{Proposition}
\newtheorem{fact}[thmctr]{Fact}

\theoremstyle{definition}
\newtheorem{definition}[thmctr]{Definition}
\newtheorem{remark}[thmctr]{Remark}

\definecolor{briancolor}{rgb}{.2, 0, 1}

\author{Constantinos Daskalakis, Gabriele Farina, Brian Hu Zhang \\ MIT  \\ \texttt{\{costis,gfarina,zhangbh\}@mit.edu}}
\date{}

\let\cite\citep
\let\citet\citep

\title{Polynomial-time computation of $\ell_p$-contraction fixed points\\for even $p$}

\begin{document}

\begin{titlepage}
\maketitle

\begin{abstract}
We give a $\poly(d, p, \log(1/\epsilon))$-time algorithm that computes an $\epsilon$-approximate fixed point of any $\ell_p$-nonexpansive map $f : \mathcal{X} \to \mathcal{X}$, where $\mathcal{X} \subset \R^d$ is a convex compact set and $p$ is an even integer. This is the first algorithm with $\poly(d, \log(1/\epsilon))$ runtime for any fixed $p \ne 2$. Our techniques are based on a computationally efficient version of Sion's theorem for non-compact minmax problems, and extend to more general total search problems that admit low-degree polynomial potentials.
\end{abstract}

\setcounter{tocdepth}{1}
{\small\tableofcontents}
\end{titlepage}

 \input{body}

\end{document}

%% file: commands.tex
\usepackage[shortlabels]{enumitem}
\usepackage[table]{xcolor}
\usepackage{mathtools, xspace, amssymb, physics2, bm, complexity, cancel, centernot} %
\PassOptionsToPackage{colorlinks, citecolor=teal, linkcolor=purple}{hyperref}
\usepackage{parskip, amssymb, anyfontsize, nicefrac, bbm}
\usepackage{comment}

\usephysicsmodule{qtext.legacy,ab}

\newcommand{\delimit}[3]{
    \NewDocumentCommand#1{sm}{
    \IfBooleanTF##1
    {\left#2 ##2 \right#3}
    {\mathopen{#2} ##2 \mathclose{#3}}%
    }
}

\delimit\ceil\lceil\rceil
\delimit\floor\lfloor\rfloor
\delimit\ip\langle\rangle
\delimit\norm\lVert\rVert
\delimit\abs\lvert\rvert

\let\op\operatorname
\let\vec\boldsymbol
\def\va{{\vec{a}}}
\def\vb{{\vec{b}}}

\def\vw{{\vec{w}}}
\def\vx{{\vec{x}}}
\def\vy{{\vec{y}}}
\def\vz{{\vec{z}}}

\let\mat\mathbf
\def\mA{{\mathbf{A}}}
\def\mB{{\mathbf{B}}}

\def\mM{{\mathbf{M}}}

\def\mP{{\mathbf{P}}}

\def\cG{{\mathcal{G}}}

\def\cQ{{\mathcal{Q}}}

\def\cX{{\mathcal{X}}}
\def\cY{{\mathcal{Y}}}

\renewcommand{\R}{\mathbb R}

\let\E\relax
\DeclareMathOperator*{\E}{\mathbb E}
\DeclareMathOperator*{\argmin}{argmin}

\newcommand{\ie}{{\em i.e.}\xspace}
\newcommand{\eg}{{\em e.g.}\xspace}

\newcommand{\supp}{\op{supp}}

\let\eps\epsilon
\let\grad\nabla

\makeatletter
\def\NAT@spacechar{~}%
\makeatother

\newclass{\UEOPL}{UEOPL}
\newclass{\SSG}{SSG}
\newclass{\CLS}{CLS}

%% file: body.tex
\section{Introduction}

Given a complete metric space $(\cX, D)$ and a function $f : \cX \to \cX$, we say that $f$ is a {\em contraction} (under the metric $D$) if there exists a constant $\gamma > 0$ such that $D(f(\vx), f(\vy)) \le (1 - \gamma) D(\vx, \vy)$ for every $\vx, \vy \in \cX$. {\em Banach's fixed point theorem}~\cite{Banach22:Operations} states that every contraction map has a unique fixed point $\vx^* \in \cX$. 

An important family of computational problems concerns efficiently {\em computing} the unique fixed point, up to some approximation. More formally, we seek algorithms that compute an {\em approximate fixed point} $\vx$, that is, a point $\vx$ with $D(\vx, f(\vx)) \le \eps$ where $\eps$ is given as input. Banach's original proof \cite[p. 160]{Banach22:Operations} gives one algorithm for this problem: simply iterating $f$, it is not difficult to show that $D(f^T(\vx), f^{T+1}(\vx)) \le \eps$ after $T \lesssim  \log(\text{diam}(\cX)/\eps)/\gamma$ steps, where $\lesssim$ hides a constant factor. However, common applications of contractions are in settings where $\gamma$ can be exponentially small, where this simple algorithm will be inefficient. Two well-known instances of this are {\em simple stochastic games}~\cite{Condon92:Complexity} and the {\sc Arrival} problem~\cite{Haslebacher25:Arrival}, which reduce to $\ell_\infty$ and $\ell_1$-contraction on $[0, 1]^d$ respectively. %

The complexity of finding a fixed point of a contraction map is a classical and well-studied problem. The fully general problem (where $\cX$ is an arbitrary metric space) is known to be $\CLS$-complete~\cite{Daskalakis18:Converse}. The case of convex compact sets $\cX \subset \R^d$ and norms $D(\vx, \vy) = \norm{\vx - \vy}_p$ has a muddier picture. For $p=2$, it is by now well known that a polynomial-time algorithm (that is, polynomial in $d$, $\log(1/\eps)$, and $\log(1/\gamma)$ for a well-described set, \eg, $\cX = [0, 1]^d$) exists~\cite{Sikorski93:Ellipsoid}. For other values of $p$, the problem is known to lie in $\CLS$~\cite{Daskalakis11:Continuous}, but otherwise little is known about the existence of polynomial-time algorithms. {\em Query}-efficient algorithms are known~\cite{Chen25:Computing,Haslebacher25:Query} for all $p$, but these are not time-efficient, as they require expensive fixed-point computations after every query. The best-known {\em time} complexity for $\ell_p$ contractions are either inverse-polynomial in $\eps$ or $\gamma$ (\eg, the Banach iteration discussed above) or exponential in $d$ (\eg, $\log^{d^{\Theta(1)}}(1/\eps)$~\cite{Fearnley20:Unique,Chen26:Quadratic,Feodorov26:Faster} for the hypercube $\cX = [0, 1]^d$). The question of whether $\poly(d, \log(1/\eps), \log(1/\gamma))$-time algorithms exist has been open for all $p \ne 2$.

\subsection{Main result and technical overview}

Our main result is an efficient algorithm for $\ell_p$-contraction in the case that $p$ is an even integer.

\begin{theorem}[Main result, informal; formal statement in \Cref{th:contraction}]\label{th:main informal}
  For any even integer $p \ge 2$, there is a $\poly(d, p, \log(1/\eps))$-time algorithm for finding  an $\epsilon$-approximate fixed point of an $\ell_p$-contractive map $f : \cX \to \cX$ on a well-described convex compact set $\cX$ (no matter the value of $\gamma$).
\end{theorem}

We first give an informal proof sketch of \Cref{th:main informal}.

We consider the following generalization of the fixed point problem: given a function $f : \cX \to \cX$ (maybe not an $\ell_p$-contraction), find a {\em distribution} on $\cX$ such that, for every $\hat\vx \in \cX$, we have
\begin{align}\label{eq:expanding}
  \E_{\vx\sim\mu} \norm{\vx - \hat\vx}_p^p \le \E_{\vx\sim\mu} \norm{f(\vx) - \hat\vx}_p^p.
\end{align}
If $f$ is in fact an $\ell_p$-contraction, taking $\hat\vx = \vx^*$ to be the fixed point of $f$, the only distribution $\mu$ satisfying \eqref{eq:expanding} is the point distribution $\mu = \delta(\vx^*)$. Thus, solving this more general problem also solves our specific problem of $\ell_p$-contraction.\footnote{In fact, if a {\em nontrivial} distribution $\mu$ satisfying \eqref{eq:expanding} exists, then it can be used to construct a certificate that $f$ is not actually contractive, and thus we can also solve the total version of the contraction problem. We will do this in \Cref{th:contr total}.} 
By convexity of the norm, and rewriting notation, it is sufficient to consider the zero-sum game
\begin{align}
  \min_{\mu \in \Delta(\cX)}\max_{G \in \cG}\E_{\vx\sim\mu} \ip{G(\vx), \vx - f(\vx)} \label{eq:minimax}
\end{align}
where
\begin{align}
  \cG = \{ \grad Q : Q \in \cQ \} \qq{and} \cQ = \{ \vx \mapsto \norm{\vx - \hat\vx}_p^p : \hat\vx \in \cX\}.
\end{align}
The reason for this notation change will become clear momentarily. 

If $f$ is a contraction, then the game \eqref{eq:minimax} has value exactly $0$: for the max-player, taking $Q^*(\vx) := \norm{\vx - \vx^*}_p^p$ to be the unique fixed point guarantees $\ip{\grad Q^*(\vx), \vx - f(\vx)} \ge Q^*(\vx) - Q^*(f(\vx)) \ge 0$ for every $\vx$, with equality if and only if $\vx = \vx^*$; and for the min-player taking $\mu = \delta(\vx^*)$ guarantees $\ip{G(\vx), \vx - f(\vx)} = 0$ for all $G$.

Our goal will be to find a min-player solution to \eqref{eq:minimax} that guarantees value $0$. To do so, we will apply the {\em ellipsoid against hope} (EAH) framework~\cite{Papadimitriou08:Computing,Farina24:Polynomial}. EAH was first introduced by \citet{Papadimitriou08:Computing} in the context of computing correlated equilibria in concise normal-form games. It has since been significantly generalized to also apply to other concepts of correlated equilibrium in games~\cite{Farina24:Polynomial,Zhang25:Learning,Anagnostides26:ComplexityCorrelated}, and has also found application to solving various problems related to {\em variational inequalities}~\cite{Zhang25:Expected,Anagnostides26:Polynomial}. It can be thought of as an efficient version of Sion's minimax theorem for when the dual problem admits an efficiently computable solution \cite{Farina26:Turning}.

Informally speaking, EAH allows us to circumvent the intractability of representing one player's strategy set (here, the fact that the min-player's strategy set is the set of distributions over an infinite set!) {\em if} we have the following properties:
\begin{enumerate}
  \item the {\em max}-player's strategy set is nicely behaved. Here, that means that $\cG$ is low-dimensional and convex, and
  \item we have a {\em semi-separation} oracle for $\cG$. That is, given some $G : \cX \to \R^d$, we need to either
  \begin{enumerate}
    \item {\em linearly separate} $G$ from $\cG$, {\em or}
    \item find an $\vx \in \cX$ ({or, more generally, a distribution $\mu \in \Delta(\cX)$, though we will not need that here}) that ``defeats'' $G$ in the sense that $\ip{G(\vx), \vx - f(\vx)} \le \eps$.
  \end{enumerate}
\end{enumerate}
The semi-separation oracle is then used as follows. Consider the dual problem 
\begin{align}
    \qq{find} G \in \cG \qq{s.t.} \ip{G(\vx), \vx - f(\vx)} \ge \eps
\end{align}
for some threshold $\eps$. The idea, then, is to use the semi-separation oracle to run the ellipsoid algorithm on this dual. The dual is infeasible, since a fixed point $\vx$ makes $\ip{G(\vx), \vx - f(\vx)} = 0$ for all $G$. Therefore, the ellipsoid algorithm eventually terminates and outputs a certificate of infeasibility. But this certificate of infeasibility is precisely a min-player strategy in \eqref{eq:minimax} that achieves value at most $\eps$!

The above idea would require $\cG$ to be nicely-behaved, that is, low-dimensional and convex. Currently, however, we have neither of these properties. To fix this, we will modify \eqref{eq:minimax} by {\em strengthening} the max-player, thus making the problem (weakly) harder, while still maintaining the property that the unique equilibrium for the min-player is $\mu = \delta(\vx^*)$. Consider instead what happens if we take $\cG$ to be the set of all {\em separable, monotone, degree-$(p-1)$} polynomial operators $G : \cX \to \R^d$, where {\em separable} means that $G(\vx)_i$ depends only on $x_i$, and {\em monotone} means
\begin{align}
  \ip{G(\vx) - G(\vy), \vx - \vy} \ge 0 \qq{for all} \vx, \vy \in \cX. \label{eq:monotone}
\end{align}
Our operator $G = \grad Q^*$ from before indeed satisfies both of these properties. Moreover, $\cG$ is now convex, because the constraints of the form \eqref{eq:monotone} are linear in $G$, and low-dimensional because we have restricted to polynomials of bounded degree and hence each $G$ can be represented by giving $O(dp)$ coefficients, one for each term $x_i^j$, for $1 \le i \le d$ and $0 \le j \le p-1$. (This is the only place where we use the assumption that $p$ is an even integer!) Therefore, we are done if we can exhibit a semi-separation oracle. To do this, we use the well-known fact that, for monotone operators $G$, there is a polynomial-time algorithm~(\eg, \cite{Nemirovski10:Accuracy}) that solves the variational inequality associated with $G$, that is, it provides a point $\vx \in \cX$ such that 
\begin{align} \label{eq:vi}
  \ip{G(\vx), \vx - \vy} \le \eps \qq{for all} \vy \in \cX,
\end{align}
and, in particular, taking $\vy = f(\vx)$, this solution $\vx$ indeed ``defeats'' $G$ in the required sense. Moreover, one can modify the algorithm so that, in the case that $G$ is not monotone, the algorithm can compute a certificate of nonmonotonicity, that is, a pair $(\vx, \vy)$ witnessing the violation of \eqref{eq:monotone}, which in turn can be used to linearly separate $G$ from the set of monotone operators. Thus, applying EAH completes the proof.

\begin{figure}[t]
\centering\scalebox{.75}{\input{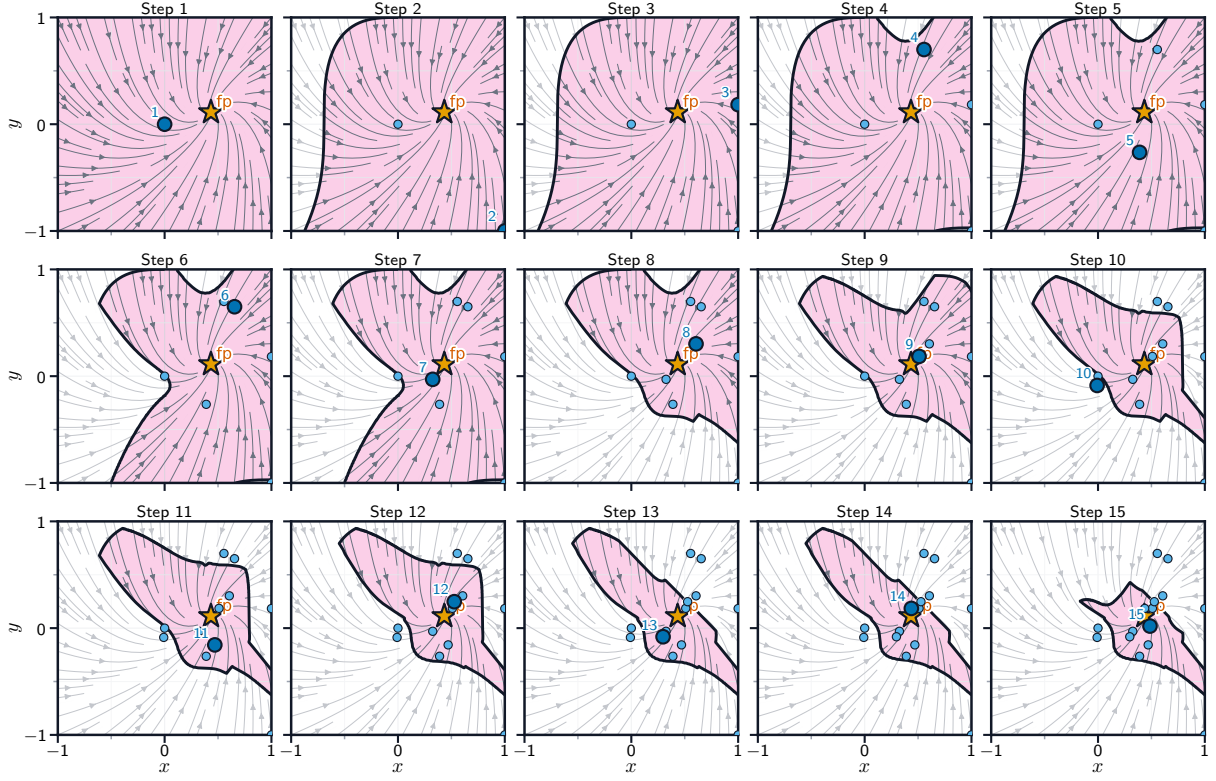}}
\caption{A two-dimensional projected visualization  of our algorithm running on a particular $\ell_4$-contraction $f : [-1, 1]^3 \to [-1, 1]^3$, for 15 iterations. The points at which $f$ is queried are marked with blue dots and labeled with the iteration numbers. On each iteration, the shaded region represents the region in which it is still possible for the fixed point to be, projected into two dimensions. The true fixed point of $f$ is marked with an orange star and labeled ``fp''. More details can be found in \Cref{sec:figure details}. The vector field represents the field $f(x, y, z) - (x, y, z)$, restricted and projected to the plane $z = z^*$ where $z^*$ is the $z$-coordinate of the fixed point.}
\label{fig:intro}
\end{figure}

\subsection{Generalizations}
We now discuss various generalizations of our main result. First, since computing approximate fixed points of {\em nonexpansive} maps (``contractions'' with $\gamma = 0$) reduces to computing approximate fixed points of contractions (\Cref{fact:nonexp}), \Cref{th:main informal} also implies that one can compute fixed points of nonexpansive maps with essentially the same time complexity, and in fact the formal statement of \Cref{th:main informal} is stated for nonexpansive maps.

The next generalization follows, in fact, from an essentially identical proof to the above, and in the technical portion of the paper we will prove the below result first and prove \Cref{th:main informal} as an application of it.

\begin{definition}[$\cG$-expected fixed point]
    Let $\cG$ be a set of operators $G : \cX \to \R^d$. A distribution $\mu \in \Delta(\cX)$ is a {$\cG$-expected fixed point} (or EFP) of $f : \cX \to \cX$ if 
    \begin{align}
        \E_{\vx\sim\mu} \ip{G(\vx), \vx - f(\vx)} \le 0 \qq{for all} G \in \cG.
    \end{align}
\end{definition}

The notion of a $\cG$-EFP generalizes prior notions of expected fixed point such as those used by \citet{Zhang25:Learning} in the context of equilibrium computation; for example, taking $G$ to be the set of all constant functions implies $\E_{\vx\sim\mu} \vx = \E_{\vx\sim\mu} f(\vx)$, hence the name {\em expected fixed point}.

To discuss computation of $\cG$-EFPs, we must first discuss how to represent the set $\cG$. Let $m : \cX \to \R^k$ be a well-behaved (\ie, bounded, Lipschitz) map. We define the set $\cG(m)$ as the set of all well-behaved {\em monotone} operators $G : \cX \to \R^d$ of the form $\vx \mapsto \mA m(\vx)$.

Our result is that it is possible to efficiently compute $\cG$-expected fixed points:
\begin{theorem}[Informal version of \Cref{th:gefp}] \label{th:gefp informal}
    There is a $\poly(d, k, \log(1/\eps))$-time algorithm that, given $\cX$, $f : \cX \to \cX$, and $m : \cX \to \R^k$ via appropriate oracles, computes an $\eps$-approximate $\cG(m)$-expected fixed point of a given $f : \cX \to \cX$.
\end{theorem}

We say a function $f$ is {\em contractive under} a convex function $Q : \cX \to \R_{\ge 0}$ if $Q(f(\vx)) \le (1-\gamma) Q(\vx)$ for all $\vx$. Then \Cref{th:gefp informal} will, after proper formalization, also imply the following result:
\begin{theorem}[Informal version of \Cref{th:Q contraction}]\label{th:Q contraction informal}
    There is a $\poly(d, k, \log(1/\eps), \log(1/\gamma))$-time algorithm that, given $\cX$, $f : \cX \to \cX$, and $m : \cX \to \R^k$ via appropriate oracles, and promised that $f$ is $(1-\gamma)$-contractive under some $Q$ with $\grad Q \in \cG$, computes an $\vx \in \cX$ for which $Q(\vx) \le \eps$.
\end{theorem}

\Cref{th:main informal} then follows from instantiating \Cref{th:Q contraction informal} with an $\ell_p$-contraction $f$, $Q(\vx) = \norm{\vx - \vx^*}_p^p$ where $\vx^*$ is the unique fixed point of $f$, and
$$m(\vx) = (1, x_1, \dots, x_d, x_1^2, \dots, x_d^2, \dots, x_1^{p-1}, \dots, x_d^{p-1}).$$
 The above result also generalizes the algorithm for fixed points of unknown Mahalanobis norms contractions of \citet{Anagnostides26:ComplexityCorrelated}, which is the case where $m = (1, \vx) \in \R^{d+1}$ (so operators $G \in \cG$ are linear).

Our last ``generalization'' is a slightly stronger form of \Cref{th:gefp informal} that only works when $\cX = [0, 1]^d$ is the hypercube and $G$ is guaranteed to be {\em separable}, that is, $G(\vx)_i$ depends only on $x_i$, but on the other hand does not require that $\cG$ consist only of monotone operators. To this end, let $\tilde\cG(m)$, for a feature map $m : [0, 1] \to \R^k$, be the set of (well-behaved, but not necessarily monotone) operators of the form $G(\vx)_i = \ip{\va_i, m(x_i)}$ for $\mA \in \R^{d \times k}$.   Then we have:
\begin{theorem}[Informal version of \Cref{th:gefp2}]
    There is a $\poly(d, k, \log(1/\eps))$-time algorithm that, given $f : [0, 1]^d \to [0, 1]^d$, and $m : [0, 1] \to \R^k$ via appropriate oracles, outputs an $\eps$-approximate $\tilde\cG(m,L)$-expected fixed point of $f$.
\end{theorem}
This result follows, essentially, because the variational inequality problem \eqref{eq:vi} is efficiently solvable when $\cX = [0, 1]^d$ and $G$ is separable, because it decomposes into $d$ independent one-dimensional VI problems, which are easy. It also gives alternative proofs and generalizations, in the obvious way, for \Cref{th:Q contraction informal} and \Cref{th:main informal}, in the special case of the hypercube. \Cref{fig:intro} shows the algorithm resulting from this analysis running on an $\ell_4$-contraction in $d=3$ dimensions.

Finally, we give a total version of \Cref{th:main informal}, in \Cref{th:contr total}. That is, in the event that $f : \cX \to \R^d$ is not a self-map or is expansive, \Cref{th:main informal} may fail, and \Cref{th:contr total} in this case outputs an explicit certificate of that failure, \ie, either a point $\vx$ with $f(\vx) \notin \cX$ or two points $\vx, \vy$ with $\norm{f(\vx) - f(\vy)}_p > \norm{\vx - \vy}_p $.

\section{Preliminaries}\label{sec:prelims}

\subsection{Notation and computational model}

$\norm{\cdot}_p$ denotes the $\ell_p$-norm on vectors and the entry-wise $\ell_p$-norm on matrices. $\Delta(\cX)$ denotes all the {\em finite-support} distributions on a set $\cX$. $B_d(\vx, r)$ is the ball of radius $r$ centered at $\vx$ in the $\ell_2$-norm in $d$ dimensions. Unless otherwise explicitly stated, ``$L$-Lipschitz'' is with respect to the $\ell_2$-norm. If $A$ and $B$ are computational problems, $A \le B$ means that there is a poly-time reduction from $A$ to $B$. The notation $\poly(a, b, \dots, \log(c, d, \dots))$ means a polynomial in $a, b, \dots, \log(c), \log(d), \dots$. The symbols $\lesssim$ and $\gtrsim$ hide absolute constants only.

We work in a black-box real-arithmetic model of computation. Basic arithmetic operations on real numbers are assumed to take unit time. To avoid numerical precision issues, we will only consider $\eps$-approximate versions of our problems of interest, for $\eps > 0$ given. 
Functions are given as evaluation oracles. Convex sets $\cX$ are assumed unless otherwise stated to be {\em well-described}, that is, 
\begin{enumerate}
  \item $\cX$ is {\em well-bounded}: $B_d(\vec x_0, 1/R) \subseteq \cX \subseteq B_d(\vec 0, R)$ where $R$ is given as input (but $\vx_0$ is not), and
  \item $\cX$ is given by an exact separation oracle, that is, an oracle that, given $\vz \in \R^d$, either declares that $\vz \in \cX$ or gives a vector $\va \in \R^d$ with $\ip{\va, \vx - \vz} > 0$ for every $\vx \in \cX$. 
\end{enumerate}
The results in this paper are not sensitive to these choices, and they are made only to ensure formality and cleanliness. For example, all the results in the paper hold, with standard minor modifications, in the standard finite-precision Turing model of computation or with approximate separation oracles.

We will use the parameter $L > 0$ to bound various numerical parameters of the problem, such as bounds on function outputs, Lipschitz constants, and so on. We will generally assume $L$ and $R$ are large and $\eps$ is small; for example, it is enough to take $L \gtrsim \sqrt{d}, R \ge 1$, and $\eps \le 1$.

Various problems throughout the paper will involve outputting {\em distributions}. In these settings, we will only work with distributions of finite support, represented as a list of support elements with their associated probabilities. Algorithms that output distributions will output them in this format.

\subsection{Contractions}

\begin{definition}
  Let $\cX \subseteq B_d(\vec 0, R)$ be a convex compact set and $f : \cX \to \cX$ be a function. We say that $f$ is a $(1-\gamma)$-contraction under a norm $\norm{\cdot} : \R^d \to \R_{\ge 0}$ if $$\norm{f(\vx) - f(\vy)} \le (1 - \gamma) \norm{\vx - \vy} \qq{for all} \vx, \vy \in \cX.$$
\end{definition}
If $f$ satisfies the above definition with $\gamma = 0$, we say that $f$ is {\em nonexpansive}, if $\gamma > 0$ then $f$ is {\em contractive}. {\em Banach's fixed point theorem} guarantees that every contractive $f : \cX \to \cX$ has a unique fixed point $\vx^* \in \cX$. An {\em $\eps$-approximate fixed point} is a point $\vx$ with $\norm{\vx - f(\vx)} \le \eps$.

\begin{remark}[Almost- and near-fixed points]
    One could also define a $\delta$-{\em near} fixed point as a point $\vx$ such that $\norm{\vx - \vx^*} \le \delta$. The two problems are polynomially equivalent for contraction maps. We state our results for $\eps$-approximate equilibria of nonexpansive maps, but the results would also apply to $\delta$-near fixed points of contractive maps, via these equivalences. For nonexpansive maps, however, near-fixed points are hard: there is not even a finite-query algorithm for computing a near-fixed point of an $\ell_\infty$-nonexpansive map. For more details on the relationship between the two notions of approximation, we refer the interested reader to \cite{Chen25:Computing}. 
\end{remark}

We start with the following straightforward fact:

\begin{fact}[\cite{Chen25:Computing}]\label{fact:nonexp}
  Let $f : \cX \to \cX$ be nonexpansive under $\norm{\cdot}$, $\eps > 0$, and let $\vx_0$ be an arbitrary point in $\cX$. Then $f_\gamma(\vx) := (1 - \gamma) f(\vx) + \gamma \vx_0$ is a $(1 - \gamma)$-contraction under $\norm{\cdot}$, and moreover every $\eps$-approximate fixed point of $f_\gamma$ is an $(\eps + \gamma \op{diam}_{\norm{\cdot}}(\cX))$-fixed point of $f$, where $\op{diam}_{\norm{\cdot}}$ is the diameter of $\cX$ under $\norm{\cdot}$.
\end{fact}
\begin{proof}
To see that $f_\gamma$ is a $(1-\gamma)$-contraction, note that
\begin{align}
    \norm{f_\gamma(\vx) - f_\gamma(\vy)} = (1 - \gamma)\norm{f(\vx) - f(\vy)} \le \norm{\vx - \vy}
\end{align}
for any $\vx, \vy \in \cX$.
For the second claim, if $\vx$ is an $\eps$-approximate fixed point, then 
\begin{align}
    \norm{f(\vx) - \vx} \le \norm{f_\gamma(\vx) - \vx} + \norm{f_\gamma(\vx) - f(\vx)} \le \eps + \gamma \norm{\vx_0 - f(\vx)} \le \eps + \gamma \op{diam}_{\norm{\cdot}}(\cX)
\end{align}
and we are done.
\end{proof}

Thus, computing fixed points of nonexpansive maps reduces to computing fixed points of contractive maps.

\subsection{Ellipsoid against hope}

In this section, we introduce {\em ellipsoid against hope}, an algorithm first introduced by \citet{Papadimitriou08:Computing} for computing normal-form games  and later generalized by \citet{Farina24:Polynomial} to situations far beyond normal-form games. Here, we state and use an even more general version.

Consider an optimization problem of the form
\begin{align}
\qq{find} \mu \in \Delta(\cX) \qq{such that} \E_{\vx\sim\mu} \ip{F(\vx), \vy} \le \epsilon \qq{for all} \vy \in \cY. \tag{EAH-P} \label{eq:eah primal}
\end{align}
where $\cY \subset \R^\ell$ is a convex compact set, $\cX$ is arbitrary, and $\norm{F(\vx)}_2 \le L$ for every $\vx$. We are given an oracle that computes $F$. We assume that \eqref{eq:eah primal} is feasible with $\eps = 0$. It follows that the following dual program is infeasible:
\begin{align}
\qq{find} \vy \in \cY \qq{such that} \ip{F(\vx), \vy} > \eps \qq{for all} \vx \in \cX. \tag{EAH-D} \label{eq:eah dual}
\end{align}
That is, for every $\vy \in \R^\ell$, either $\vy \notin \cY$, or there is $\vx \in \cX$ with $\ip{F(\vx), \vy} \le \epsilon$. The goal of EAH is essentially to give a computational converse to this weak duality statement. That is, assuming only a computational proof of the above dual infeasibility statement, solve \eqref{eq:eah primal}. The computational proof of infeasibility is called a {\em semi-separation oracle}:
\begin{definition}
  A {\em semi-separation oracle} for \eqref{eq:eah primal} is an oracle that, given $\vy \in \R^\ell$ and $\epsilon > 0$, outputs either
  \begin{enumerate}
    \item (Good enough response) a distribution $\mu \in \Delta(\cX)$ with $\E_{\vx\sim\mu} \ip{F(\vx), \vy} \le \epsilon$, or
    \item (Separation) a direction $\va \in \R^\ell$ such that $\ip{\va, \vy' - \vy} > 0$ for all $\vy' \in \cY$. 
  \end{enumerate}
\end{definition}
The name ``good enough response'' is due to \citet{Farina24:Polynomial}, and comes from the fact that such a distribution is not a best response, $\argmin_{\vx \in \cX} \ip{F(\vx), \vy}$, but only ``good enough'' in the sense that it proves that $\vy$ does not achieve value higher than $\epsilon$. Moreover, notice that the semi-separation oracle does not need to necessarily know whether $\vy \in \cY$; it is valid for a semi-separation oracle to output a good-enough response even when $\vy \notin \cY$. Both these facts are crucial: indeed, it is possible for semi-separation oracles to exist even when the set $\cY$ has no efficient {\em separation} oracle and even when true best responses are hard to compute. Many applications of EAH, including ours and many of the papers previously cited, will make use of this distinction. 

The main result of EAH is the following.
\begin{theorem}[Ellipsoid against hope, generalized form of \citet{Farina24:Polynomial}]\label{th:eah} Suppose $\cY$ is well-bounded, and assume that a semi-separation oracle exists for \eqref{eq:eah primal}. Then a solution to \eqref{eq:eah primal} can be solved in time $\poly(\ell, \log(R, L, 1/\eps))$, and the solution can be ensured to have support at most $\ell+1$.
\end{theorem}
For completeness, we include a proof in \Cref{sec:proofs prelims}.

\subsection{Variational inequalities}

In this subsection, we give required background on efficient algorithms for {\em variational inequality problems}.
\begin{definition}
  The $\eps$-approximate {\em variational inequality} (VI) problem is the following: given feasible set $\cX$ and an operator $G : \cX \to \R^d$, find a point $\vx \in \cX$ with
  \begin{align}
  \ip{G(\vx), \vx - \vy} \le \eps \qq{for all} \vy \in \cX.
  \end{align}
\end{definition}
We will primarily be interested in {\em monotone} VIs, that is, VIs where the operator satisfies \eqref{eq:monotone}. The following essentially known result about the complexity of monotone VIs will be crucial in our analysis.
\begin{theorem}[Complexity of monotone variational inequalities]\label{th:monotone}
  There is a $\poly(d, \log(R, L, 1/\eps))$-time algorithm that, given parameters $R, L,\eps > 0$, convex compact set $\cX$, and $L$-Lipschitz operator $G : \cX \to B_d(\vec 0, L)$, outputs {\em either}
  \begin{enumerate}
    \item an $\eps$-approximate VI solution, that is, a point $\vx \in \cX$ with  $\ip{G(\vx), \vx - \vy} \le \eps$ for all $\vy \in \cX$, {\em or}
    \item a certificate of nonmonotonicity, that is, two points $\vx, \vy \in \cX$ with $\ip{G(\vx) - G(\vy), \vx - \vy} < 0$.
  \end{enumerate}
\end{theorem}
For completeness, two different proofs are given in \Cref{sec:proofs prelims}.
The fact that $G$ is not required to be monotone in the above theorem will be crucial. Indeed, we will later invoke the above result with operators $G$ for which we do not know whether or not they are monotone, where the above result will be used as a semi-separation oracle!

\section{Generalized expected fixed points and main result}\label{sec:main}

We will start by defining a generalization of the norm contraction fixed point problem. 

\begin{definition}[$\cG$-expected fixed point]
    Let $\cG$ be a set of operators $G : \cX \to \R^d$. A distribution $\mu \in \Delta(\cX)$ is an {\em  $\eps$-approximate $\cG$-expected fixed point} (or EFP) of $f$ if 
    \begin{align}
        \E_{\vx\sim\mu} \ip{G(\vx), \vx - f(\vx)} \le \eps \qq{for all} G \in \cG.
    \end{align}
\end{definition}

We will phrase our next result in terms of $\cG$-expected fixed points. To do this, we must first discuss how to represent the set $\cG$. Let $m : \cX \to B_k(\vec 0, L)$ be an arbitrary $L$-Lipschitz map. For cleanliness, we will assume $L$ is sufficiently large (say, $L \gtrsim \sqrt{d}$ is enough) that $k \ge d+1$, and the first $d+1$ coordinates of $m(\vx)$ are the constant $1$ followed by $\vx$ itself. We define the set $\cG(m,L)$ as the set of all {\em monotone} operators $G : \cX \to \R^d$ of the form $\vx \mapsto \mA m(\vx)$, where $\norm{\mA}_2 \le L$. For notational simplicity, we also define $\cY(m,L)$ to be the set of all matrices $\mA$ with $\norm{\mA}_2 \le L$ such that $\vx \mapsto \mA m(\vx)$ is monotone, so that we have
\begin{align}
  \cG(m,L) = \{ \vx \mapsto \mA m(\vx) : \mA \in \cY(m,L) \}.
\end{align}
We now discuss how to represent $\ell_p$-contractions in this language.

\begin{proposition}\label{prop:m is nice}
  Let $\cX \subseteq B_d(\vec 0, R)$, and $p$ be a positive even integer. Let  $m : \cX \to \R^k$ where $k = d(p-1) + 1$ be given by $$m(\vx) = (1, x_1, \dots, x_d, x_1^2, \dots, x_d^2, \dots, x_1^{p-1}, \dots, x_d^{p-1}).$$ Then for $L \lesssim \poly(d, p) \cdot (2R)^p$, we have that $m$ is $L$-Lipschitz and $L$-bounded, and moreover for every $\hat\vx \in \cX$ we have $\grad Q^{\hat\vx} \in \cG(m,L)$, where $Q^{\hat\vx}(\vx) := \norm{\vx - \hat\vx}_p^p$.
\end{proposition}
The proof is straightforward, and deferred to \Cref{sec:proofs main}, together with a version that also includes cross-terms which will be useful later. Crucially, the function $Q$ above is a {\em polynomial} when $p$ is an even integer---this is the reason why our techniques work for even $p$ and not for other values of $p$.

\begin{theorem}\label{th:gefp}
  There is a $\poly(d, k, \log(R, L, 1/\eps))$-time algorithm that, given parameters $R, L, \eps > 0$, well-described $\cX$, function $f : \cX \to \cX$, and $L$-Lipschitz feature map $m : \cX \to B_k(\vec 0, L)$, outputs an $\eps$-approximate $\cG(m,L)$-expected fixed point of $f$, of support at most $kd + 1$.
\end{theorem}
\begin{proof}
  We wish to solve the problem
    \begin{align}
        \qq{find} \mu \in \Delta(\cX) \qq{s.t.} \E_{\vx\sim\mu} \ip{G(\vx), \vx - f(\vx)} \le \eps \qq{for all} G \in \cG(m,L), 
    \end{align}
    or equivalently, 
  \begin{align}
        \qq{find} \mu \in \Delta(\cX) \qq{s.t.} \E_{\vx\sim\mu} \ip{\mA m(\vx), \vx - f(\vx)} \le \eps \qq{for all} \mA \in \cY(m,L), 
    \end{align}
  We will use EAH, so it suffices to check the conditions of \Cref{th:eah}.
   To see that $\cY(m,L)$ is well-bounded, consider the map
  $\vx \mapsto \vx + \mB m(\vx)$. Then we have
  \begin{align}
    \ip{\vx - \vy + \mB (m(\vx) - m(\vy)), \vx - \vy}
    &= \norm{\vx - \vy}_2^2 + \ip{\mB (m(\vx) -m(\vy)), \vx - \vy}
    \\&\ge \norm{\vx - \vy}_2^2 - L \norm{\mB}_2 \norm{\vx - \vy}_2^2,
  \end{align}
so the map is monotone as long as $\norm{\mB}_2 \le 1/L$, that is $\cY(m,L)$ contains a ball of radius $1/L$ around the identity map, which is in $\cY(m,L)$ because it is represented by a matrix $\mA$ with $\norm{\mA}_2 = \sqrt{d} \le L$.

  For a semi-separation oracle,  given $G(\vx) = \mA m(\vx)$ for $\mA \in \R^{d \times k}$ with $\norm{\mA}_2 \le L$, find either an $\vx$ with $\ip{G(\vx) ,\vx - f(\vx)} \le \epsilon$, or $\vx, \vy$ showing that $G$ is not monotone. But \Cref{th:monotone} is exactly such an oracle: indeed, if it returns an $\eps$-VI solution $\vx$, then $\ip{\mA m(\vx), \vx - f(\vx)} \le \eps$ (because in particular $f(\vx) \in \cX$); if it returns a monotonicity violation $(\vx, \vy)$ then $\ip{\mA(m(\vx) - m(\vy)), \vx - \vy} < 0$. 
\end{proof}

\begin{remark}[Existence of expected fixed points]
    The above theorem shows, by construction, that $\eps$-approximate $\cG$-EFPs always exist. There are, however, more direct proofs that $\cG$-EFPs always exist. For more on this, see \Cref{app:existence of EFPs}.
\end{remark}

We are now ready to state and prove the main result.
\begin{theorem}\label{th:contraction}
    There is a $\poly(d, p, \log(R, 1/\eps))$-time algorithm that, given parameters $R, \epsilon > 0$, positive even integer $p$, well-described $\cX$, and function $f : \cX \to \cX$ that is nonexpansive under the $\ell_p$-norm, computes an $\eps$-approximate fixed point of $f$.
\end{theorem}
\begin{proof}
  By \Cref{fact:nonexp}, we may WLOG set $f \gets f_\gamma$, which is guaranteed to be $(1-\gamma)$-contractive, for $\gamma$ to be chosen later.  By \Cref{th:gefp} with $m, L$ set as per \Cref{prop:m is nice}, we can compute an $\eps'$-approximate $\cG(m, L)$-expected fixed point of $f$ in time $\poly(d, p, \log(1/\eps'))$. Now taking $Q(\vx) := \norm{\vx - \vx^*}_p^p$ where $\vx^*$ is the unique fixed point of $f$, we have $\grad Q \in \cG(m, L)$, and therefore
  \begin{align}
        \eps' &\ge \E_{\vx\sim\mu} \ip{\grad Q(\vx), \vx - f(\vx)}
        \ge \E_{\vx\sim\mu} [Q(\vx) - Q(f(\vx))]
        \ge \gamma \E_{\vx\sim\mu} Q(\vx)
  \end{align}
  where we use, in turn: the definition of expected fixed point, convexity of $Q$, and finally contractivity of $f$, namely, 
  \begin{align}
    \norm{f(\vx) - \vx^*}_p^p \le (1 - \gamma)^p \norm{\vx - \vx^*}_p^p \le (1 - \gamma) \norm{\vx - \vx^*}_p^p.
  \end{align}
  
  Thus, there must be an $\vx \in \supp\mu$ with 
  \begin{align}
    \norm{\vx - \vx^*}_p \le \delta := \ab(\frac{\eps'}{\gamma})^{1/p}
  \end{align}
  and therefore
  \begin{align}\label{eq:eps bound}
    \norm{\vx - f_\gamma(\vx)}_p \le \underbrace{\norm{\vx - \vx^*}_p}_{\le \delta} + \underbrace{\norm{\vx^* - f_\gamma(\vx^*)}_p}_{=0} + \underbrace{\norm{f_\gamma(\vx^*) - f_\gamma(\vx)}_p}_{\le \norm{\vx^* - \vx} \le \delta} \le 2 \cdot \ab(\frac{\eps'}{\gamma})^{1/p},
  \end{align}
  which implies, via \Cref{fact:nonexp}, that $\norm{\vx - f(\vx)}_p \le 2 ( ({\eps'}/{\gamma})^{1/p} + \gamma R) $. Taking $\gamma = \eps/(4R)$ and $\eps' = \gamma\cdot (\eps/4)^p$ completes the proof.
\end{proof}

\section{Generalizations}\label{sec:generalizations}

In this section, we discuss generalizations of $\ell_p$-contraction problems whose efficient solutions are enabled by our techniques.

\subsection{Unknown norms, rotated norms, and generalized contractions}

In this section, we discuss a generalization of \Cref{th:contraction} to arbitrary norms, and indeed to contractions under functions that may not be norms.
\begin{definition}
  A function $f : \cX \to \cX$ is a $(1-\gamma)$-contraction under a convex function $Q : \cX \to \R_{\ge 0}$ if $$Q(f(\vx)) \le (1 - \gamma) Q(\vx) \qq{for all} \vx \in \cX.$$ 
\end{definition}

 The notion of a contraction under $Q$ generalizes the norm contraction, in the following sense:
\begin{fact}\label{fact:norm is poly}
  Suppose that $f : \cX \to \cX$ is a $(1-\gamma)$-contraction under the $\ell_p$-norm, where $p$ is a positive even integer. Then $f$ is a $(1-\gamma)^p \le (1 - \gamma)$-contraction under the degree-$p$ polynomial $Q(\vx) = \norm{\vx - \vx^*}_p^p$, where $\vx^*$ is the unique fixed point of $f$.
\end{fact}

Our main technical result, \Cref{th:gefp}, can then be used to prove a more general version of \Cref{th:contraction}. Indeed, let $f$ be a $(1-\gamma)$-contraction under convex function $Q$. Suppose $\grad Q \in \cG$, and let $\mu$ be an $\eps$-approximate $\cG$-EFP of $f$. Then \begin{align}
        \eps &\ge \E_{\vx\sim\mu} \ip{\grad Q(\vx), \vx - f(\vx)}
        \ge \E_{\vx\sim\mu} [Q(\vx) - Q(f(\vx))]
        \ge \gamma \E_{\vx\sim\mu} Q(\vx)
    \end{align}
and therefore we have
\begin{align}
  Q(\vx^*) \le \E_{\vx\sim\mu} Q(\vx) \le \frac{\eps}{\gamma}.
\end{align}
with $\vx^* := \E_{\vx\sim\mu} \vx$.
Hence, \Cref{th:gefp} implies the following result:
\begin{theorem} \label{th:Q contraction}
  There is a $\poly(d, k, \log(R, L, 1/\eps, 1/\gamma))$-time algorithm that, given parameters $R, L, \eps$, and $\gamma > 0$, well-described $\cX$, function $f : \cX \to \cX$, and $L$-Lipschitz feature map $m : \cX \to B_k(\vec 0, L)$, and promised that $f$ is $(1-\gamma)$-contractive under an (unknown) convex $Q : \cX \to \R_{\ge 0}$ such that $\grad Q \in \cG(m, L)$, outputs a point $\vx^*$ with $Q(\vx^*) \le \eps/\gamma$.
\end{theorem}

It also generalizes prior notions of expected fixed point that have been used for computing equilibria in games: indeed, the {\em expected fixed point}~\cite{Zhang24:Efficient}, {\em quadratic expected fixed point}, and {\em convex expected fixed point} \cite{Anagnostides26:ComplexityCorrelated} arise respectively by setting $\cG$ to the set of all constant operators, affine operators $\vx \mapsto \mA \vx + \vb$ with $\mA \succeq \vec 0$, and gradients of convex functions.

\paragraph{Unknown norms} The proof of \Cref{th:contraction} can also be applied, essentially verbatim, to the case where the value $p$ is not known a priori. Indeed, we can simply run the algorithm for increasing values of $p$, and target accuracy $\eps/\sqrt{d}$ until we find a point $\vx$ with $\norm{\vx - f(\vx)}_2 \le \eps$, which guarantees $\norm{\vx - f(\vx)}_p \le \eps$ for every $p \ge 2$.

\paragraph{Other norms, such as rotated norms} \Cref{th:gefp} can also be used to find fixed points of contractions under other norms. As an example, consider the seminorm $\vx \mapsto \norm{\mP \vx}_p$ for some matrix $\mP$ that is promised to be bounded, say, $\norm{\mP}_2 \le L$. Then, as before, $\vx \mapsto \norm{\mP \vx}_p^p$ is a polynomial of degree $p$ with bounded coefficients, but since it now contains cross-terms, our function $m$ must now include such cross-terms as well. Combining with the above observation of unknown norms, we have the following generalization.
\begin{theorem} \label{th:rotated norm}
  Let $f : \cX \to \cX$ be a function, promised to be nonexpansive under some norm $\norm{\vx} := \norm{\mP \vx}_q$ with $\norm{\mP}_2 \le 1$ and $q \le p$ a positive even integer. Then there is a $\poly(d^p, \log(R, 1/\eps))$-time algorithm that, given $\cX, f$, and $\eps > 0$ (but not $q$ or $\mP$), computes a distribution $\mu$ with $\E_{\vx\sim\mu} \norm{\vx^* - f(\vx^*)}^q \le \eps.$ If $\mP$ is promised to be diagonal then the runtime reduces to $\poly(d, p, \log(R, 1/\eps))$.
\end{theorem}
\begin{proof}
  Follow the proof of \Cref{th:contraction}. If $\mP$ is not promised to be diagonal, we can take $m(\vx)$ as in \Cref{prop:m is nice with cross}. Otherwise, we take $m(\vx)$ as in \Cref{prop:m is nice}. %
\end{proof}
If $q$ and $\mP$ are given, an approximate fixed point $\vx$ can, as usual, be extracted by iterating over $\mu$ and checking every $\vx$ in its support. Otherwise, a randomized algorithm with arbitrarily high success probability $1-\tau$ can be created in the standard manner: multiply $\eps$ by a factor of $\tau$, run the algorithm, and then sample from the resulting $\mu$.

\subsection{The special case of $\ell_p$-norms on the hypercube}\label{sec:lp-simplified}

In the special case that the desired norm is the $\ell_p$-norm and the domain is the hypercube $\cX = [0, 1]^d$, our algorithm can be simplified and in fact we can show a slightly more general result. This is because, for the special case of separable VI over the hypercube, VI solutions can be found efficiently even when the operator is not monotone. We say a VI operator $G : [0, 1]^d \to \R^d$ is {\em separable} if $G(\vx)_i$ depends only on $x_i$.
\begin{proposition}[Complexity of separable VIs]\label{th:separable VI}
  There is a $\poly(d, \log(L, 1/\eps))$-time algorithm that, given parameters $L,\eps > 0$, and $L$-Lipschitz separable operator $G : [0, 1]^d \to B_d(\vec 0, L)$, outputs an $\eps$-approximate VI solution, that is, a point $\vx \in \cX$ with  $\ip{G(\vx), \vx - \vy} \le \eps$ for all $\vy \in [0, 1]^d$.
\end{proposition}

\begin{proof}
  Let $G(\vx)_i = G_i(x_i)$. For each $i \in [d]$, there are three cases:
  \begin{enumerate}
    \item $G_i(0) \ge 0$. Then setting $x_i = 0$ gives $G_i(x_i)(x_i-y_i) \le 0$ for all $y_i \in [0, 1]$.
    \item $G_i(1) \le 0$. Then setting $x_i = 1$ gives $G_i(x_i)(x_i-y_i) \le 0$ for all $y_i \in [0, 1]$.
    \item $G_i(0) < 0$ and $G_i(1) > 0$. Then a binary search, running in $\polylog(L, 1/\eps)$ time, can be used to identify an $x_i$ such that $|G_i(x_i)| \le \eps/d$, so in particular $G_i(x_i) (x_i - y_i) \le \eps/d$ for all $y_i \in [0, 1]$ as well.
  \end{enumerate}
  Construct $\vx \in [0, 1]^d$ by selecting $x_i$ as above for every coordinate $i$; this takes time $\poly(d, \log(L, 1/\eps))$. Then for any coordinate $i$ we have $G_i(x_i) (x_i - y_i) \le \eps/d$ for all $y_i \in [0, 1]$; summing over $i \in [d]$ yields the theorem.
\end{proof}

This generalization allows us to drop the condition of monotonicity from \Cref{th:gefp}, for the special case that $\cX$ is the hypercube and $G$ is separable. In particular, for any $m : [0, 1] \to B_d(\vec 0, L)$ and $L > 0$, define $\tilde\cY(m, L)$ to be the set of matrices $\mA \in \R^{d \times k}$ with $\norm{\mA}_2 \le L$, and $\tilde\cG(m, L)$ to be the set of all (not necessarily monotone) operators of the form $G(\vx)_i = \ip{\va_i, m(x_i)}$, for $\mA \in \tilde\cY(m, L)$, where $\va_i$ is the $i$th row of $\mA$.

\begin{theorem}\label{th:gefp2}
  There is a $\poly(d, k, \log(L, 1/\eps))$-time algorithm that, given parameters $L, \eps > 0$, function $f : [0, 1]^d \to [0, 1]^d$, and $L$-Lipschitz feature map $m : [0, 1] \to B_k(\vec 0, L)$, outputs an $\eps$-approximate $\tilde\cG(m,L)$-expected fixed point of $f$, of support at most $kd + 1$.
\end{theorem}
\begin{proof}
  Follow the proof of \Cref{th:gefp}, with $G(\vx)$ set, as above, by $G(\vx)_i = \ip{\va_i, m(x_i)}$.
  The well-boundedness of $\tilde\cY(m, L)$ now comes for free, because the set of valid matrices $\mA$ is by definition a ball. For a semi-separation oracle, use \Cref{th:separable VI}, noting that nonmonotone operators are valid and thus the ``separation'' part is simply a separation oracle for the ball $\norm{\mA}_2 \le L$.
\end{proof}
\Cref{th:gefp2} is in general incomparable to \Cref{th:gefp}: \Cref{th:gefp} requires monotonicity whereas \Cref{th:gefp2} does not, but on the other hand \Cref{th:gefp2} requires that $\cX = [0, 1]^d$ and $G$ is separable, whereas \Cref{th:gefp} does not.

The use of \Cref{th:gefp} in \Cref{th:contraction} can be replaced with \Cref{th:gefp2}, because for $Q(\vx) := \norm{\vx - \vx^*}_p^p$, the gradient $\grad Q$ is indeed separable in the required sense. Hence, this gives an alternative---and, in fact, slightly simpler---proof of our main result, but only in the special case $\cX = [0, 1]^d$.

\begin{remark}
  One may ask whether the above argument can be used to prove our main result in full generality, by extending a given $\ell_p$-nonexpansive $f : \cX \to \cX$ for $\cX \subseteq [0, 1]^d$ to an $\ell_p$-nonexpansive $\tilde f : [0, 1]^d \to [0, 1]^d$. For $p=2$, this can be done by taking $\tilde f(\vx) := f(\Pi_\cX(\vx))$ where $\Pi_\cX$ is the (Euclidean) projection onto $\cX$. However, for other values of $p$, we are not aware of any means by which to efficiently construct such an extension $\tilde f$. For $p=\infty$, $\tilde f$ can be constructed by McShane's theorem~\cite{McShane34:Extension} applied coordinate-wise, namely
  \begin{align}
      \tilde f(\vx)_i = \Pi_{[0, 1]} \ab\{ \min_{\vy \in \cX} \ab\{ f_i(\vy) +  \norm{\vx - \vy}_\infty \} \}
  \end{align}
  However, this $\tilde f$ cannot be efficiently computed, as it requires optimizing possibly nonconvex functions over $\cX$.
\end{remark}

\subsection{Total function problem version}
We now extend \Cref{th:contraction} so that it explicitly outputs certificates of promise violations in the case that the contraction or feasibility ($f(\cX) \subseteq \cX$) promises are broken. For simplicity, we still keep the well-boundedness of $\cX$ as a true promise. Otherwise, the below theorem makes the problem total.
\begin{theorem}\label{th:contr total}
  There is a $\poly(d, p, \log(R, 1/\eps))$-time algorithm that, given parameter $\epsilon > 0$, positive even integer $p$, well-described $\cX$, and function $f : \cX \to \R^d$, outputs one of the following:
  \begin{itemize}
    \item (Solution) an $\eps$-approximate fixed point of $f$,
    \item (Contraction violation) two points $\vx, \vy$ with $\norm{f(\vx) - f(\vy)}_p > \norm{\vx - \vy}_p$, or
    \item (Feasibility violation) a point $\vx \in \cX$ with $f(\vx) \notin \cX$.
  \end{itemize}
\end{theorem}
\begin{proof}
  As before, we may WLOG set $f \gets f_\gamma$ where $f_\gamma$ is as in \Cref{fact:nonexp}; in particular, $f_\gamma$ is $(1-\gamma)$-contractive if $f$ is nonexpansive, and a violation of $(1-\gamma)$-contractivity of $f_\gamma$ immediately gives a violation of nonexpansiveness of $f$.
  \Cref{th:gefp} can be straightforwardly modified to return a feasibility violation if it finds one. The main task in this proof is to circumvent the need to refer, in the proof of \Cref{th:contraction}, to know the actual fixed point $\vx^*$. Hence, let $\mu$ be an $\eps'$-approximate $\cG(m,L)$-expected fixed point of $f$, with $m$ and $L$ chosen as per \Cref{prop:m is nice}, and assume that such $\mu$ was found without any promise violations. We first check whether any pair $\vx, \vy \in \supp \mu$ violates the contraction property; if so, we output it. It remains to show that, otherwise, $\supp\mu$ must contain an approximate fixed point. 

  First, let $Q : \cX \to \R$ be any convex function with $\grad Q \in \cG(m,L)$. From the proof of \Cref{th:contraction}, we have
  \begin{align}
        \E_{\vx\sim\mu} [Q(\vx) - Q(f(\vx))] \le \eps'.
    \end{align}
  In particular, if $Q(\vx) := \norm{\vx - \vx^*}_p^p$, then this reduces to
  \begin{align}
    \E_{\vx\sim\mu} \norm{\vx - \vx^*}_p^p \le \E_{\vx\sim\mu} \norm{f(\vx) - \vx^*}_p^p + \eps' \label{eq:sfp}
  \end{align}
  for every $\vx^* \in \cX$ (not necessarily the fixed point).
  Therefore,
  \begin{align}
    \E_{\vx, \vy \sim \mu} \norm{\vx - \vy}_p^p &\le \E_{\vx, \vy \sim \mu} \norm{f(\vx) - \vy}_p^p + \eps' 
    \\&\le \E_{\vx, \vy \sim \mu} \norm{f(\vx) - f(\vy)}_p^p + 2 \eps' \label{eq:identical}
    \\&\le (1 - \gamma) \E_{\vx, \vy \sim \mu} \norm{\vx - \vy}_p^p + 2 \eps'
  \end{align}
  where the first two inequalities follow from \eqref{eq:sfp}, and the last inequality is contractivity of $f$ on the support of $\mu$. Thus, 
  \begin{align}
    \E_{\vx, \vy \sim \mu} \norm{\vx - \vy}_p^p \le \frac{2\eps'}{\gamma} \qq{and} \E_{\vx, \vy \sim \mu} \norm{f(\vx) - \vy}_p^p \le (1 - \gamma) \E_{\vx, \vy \sim \mu} \norm{\vx - \vy}_p^p + 2 \eps' \le \frac{4\eps'}{\gamma}.
  \end{align}
  Therefore,
  \begin{align}
    \E_{\vx\sim\mu} \norm{\vx - f(\vx)}_p^p &\le \E_{\vx,\vy\sim\mu} \ab(\norm{\vx - \vy}_p + \norm{f(\vx) - \vy}_p)^p
    \\&\le 2^{p-1} \E_{\vx,\vy\sim\mu} \ab(\norm{\vx - \vy}_p^p + \norm{f(\vx) - \vy}_p^p) \\&\le 2^{p} \cdot \frac{3\eps'}{\gamma},
  \end{align}
  so there is $\vx \in \supp\mu$ with 
  \begin{align}
    \norm{\vx - f(\vx)}_p \le 2 \cdot \ab(\frac{3\eps'}{\gamma})^{1/p}.
  \end{align}
  Thus, again following \Cref{th:contraction}, taking $\gamma = \eps/(4R)$ and $\eps' = (\eps/4)^p \gamma/3$ completes the proof. 
\end{proof}

\section{Discussion, Conclusion, and Future Directions}

The complexity of computing fixed points of $\ell_p$-contractions beyond the case $p=2$ has been a long-standing open problem. We have made significant progress on this problem, by giving an efficient algorithm that works for all even integers $p$. Our techniques draw from recent work on computational game theory, especially semi-separation and ellipsoid against hope, techniques typically used to compute various notions of correlated equilibrium. 

\subsection{Possibility of generalization beyond even integers}

One immediate question is whether our techniques can also be used to give efficient algorithms for values of $p$ other than even integers. We believe that, if this were possible, it would require a significantly new insight. We now discuss a few simple-sounding paths toward generalization, and why they do not work.

\paragraph{Approximating the $\ell_\infty$-norm by an $\ell_p$-norm} One may attempt generalize our results to $p = \infty$ by simply taking $p$ so large that $\norm{\cdot}_p \approx \norm{\cdot}_\infty$, and then running our algorithm on the $\ell_p$-norm. This, however, does not work: the multiplicative error induced by this approximation is on the order of $d^{1/p}$, so for a $(1-\gamma)$-contraction under $\norm{\cdot}_\infty$ to remain nonexpansive under $\norm{\cdot}_p$, we need to take $p \gtrsim \log(d)/\gamma$, which is too large since our algorithms depend polynomially on $p$. 

\paragraph{Approximating the $\ell_p$-distances for other values of $p$ using few dimensions} Our main result essentially works because, when $p$ is an even integer, the set of functions $\cQ_p := \{ \vx \mapsto \norm{\vx - \vx^*}_p^p : \vx^* \in \cX \} \subset \R^\cX$ are low-degree polynomials. One may ask whether this is true, even approximately, for other values of $p$. 

The answer to this is negative. For odd $p$, the set of maps $\mathcal{G}=\{\nabla_\vx \|\vx-\vx^*\|_p^p:\vx\in[-1,1]^d\}$ is not finite-dimensional. As an illustration, consider $d=p=1$, in which case $\mathcal{G}$ consists of sign functions $\text{sgn}(x-x^*)$, which are known to not be finite dimensional. Therefore, the analogue of \Cref{prop:m is nice} would not hold for this $\mathcal{G}$ for any finite value of $k$. For approximations, \citet{Bernstein14:Meilleure} showed that approximating the absolute value map to uniform error $\eps$ requires degree $r \gtrsim 1/\eps$, which again is not good enough since we desire $\polylog(1/\eps)$-time algorithms. Other values of $p$ run into similar barriers due to the need to represent the absolute value function.

\paragraph{Kernel methods} Our method of introducing a feature map $m : \cX \to \R^k$ may remind the reader of reproducing kernel Hilbert spaces (RKHSs). While these are a powerful tool, and have indeed been used before for learning concepts related to semi-separation and correlated equilibrium (see \eg~\citet{Farina26:Efficient}), they seem difficult to apply to our setting. Indeed, one barrier is that we rely fundamentally on cutting-plane methods whose runtime depends polynomially on the dimension $k$ of the feature space, rendering computation problematic when $k$ is large or infinite.

\subsection{Semi-separation and the usefulness of promise violations}

Search problems that have promises are often formulated as {\em total problems} rather than explicit promise problems, in which promise violations are made explicit. Consider, for example, the total form (\Cref{th:contr total}) of our main result (\Cref{th:contraction}). One reason to do this is because {\em complexity classes} of search problems are traditionally defined as total search problems. However, to our knowledge, such promise-violation formulations have mostly been regarded as, essentially, formal necessities to satisfy the $\TFNP$ definition, and nothing more. 

In this language, one way of viewing semi-separation is that it makes {\em explicit, nontrivial} use of the promise violation: that is, \Cref{th:gefp}, even if its input satisfies all its promises, can make a call to \Cref{th:monotone} with an operator $G$ that may {\em not} be monotone, and subsequently {\em use} the monotonicity violation as part of the algorithm. Is this general principle or idea, of making nontrivial use of promise violations of $\TFNP$ problems, useful in other problems as well, beyond equilibrium computation in games and our application to contractions? We leave this to future work.

\subsection{Other future directions}

We also leave several other possible directions and questions for future research.

\begin{itemize}
  \item Our methods rely fundamentally on cutting-plane algorithms, namely, the ellipsoid method. Our main results actually require {\em two} nested layers of EAH: the first to compute the correct notion of an expected fixed point (\Cref{th:gefp}), and the inner layer within \Cref{th:gefp} to semi-separate over the class of monotone operators (\Cref{th:monotone}). A pressing question is how to make these algorithms more efficient in practice. One resolution would be to give a more efficient alternative to ellipsoid against hope itself: that would speed up both layers.
  \item What reductions are there among the problems of $\ell_p$-contraction, beyond the case of even integer $p$? For example, are all $\ell_p$-contractions for $p$ other than even integers computationally equivalent? Can any be reduced to (say) $\ell_4$-contraction, and thus be solved efficiently?
\end{itemize}

\section*{Acknowledgements}

C.D. was supported by a Simons Investigator Award, a Simons Collaboration on Algorithmic Fairness, ONR MURI grant N00014-25-1-2116, ONR grant N00014-25-1-2296. G.F. was supported in part by the National Science Foundation awards CCF-2443068 and IIS-2552046, the Office of Naval Research grant N000142512296, and an AI2050 Early Career Fellowship.

\section*{AI Disclosure}

Some routine derivations were delegated to a standard LLM, and subsequently hand-verified and completely re-written by the authors. All writing, as well as all the main ideas of the paper, is by the authors. The authors maintain responsibility for all content of the paper. An LLM was also used in generating \Cref{fig:intro}, but the authors manually checked the correctness of the figure.

\bibliography{dairefs}

\appendix

\section{Details on \Cref{fig:intro}}
\label{sec:figure details}

\paragraph{Notation} For this section and this section only, $\norm{\cdot}_p$ on matrices refers to the operator $p$-norm.

The function $f$ in \Cref{fig:intro} is given by 
\begin{align}
    f(\vx) = \mM \op{tanh}(\vx) + \vb
\end{align}
where $\op{tanh}$ is meant coordinate-wise, and $\mM, \vb$ are
\begin{align}
    \mM = \begin{pmatrix*}[r]
0.44274663 & -0.34622668 & 0.05119011 \\
0.59488020 & -0.12626562 & 0.10124367 \\
-0.64535783 & 0.10505465 & 0.07918490
\end{pmatrix*}, \qq{} \vb = \begin{pmatrix*}[r]
0.3 \\
-0.1 \\
0.1
\end{pmatrix*}.      
\end{align}
It is straightforward to check by inspection of the rows of $\mM$ that $f([-1,1]^3) \subseteq [-1,1]^3$.

\begin{lemma}
    $f$ is not $\ell_2$-contractive.
\end{lemma}
\begin{proof}
The Jacobian of $f$ is given by
\begin{align}\label{eq:example jacobian}
    Df(\vx) &= \mat M \op{diag}(\op{sech}^2(\vx)),
\end{align}
so     $Df(\vec 0) = \mat M$. But $\norm{\mat M}_2 \approx 1.031 > 1$, so $f$ is not $\ell_2$-contractive.
\end{proof}

\begin{lemma}
    $f$ is $\ell_4$-contractive, with $\ell_4$-Lipschitz constant at most $0.94$ and unique fixed point $\vx^* \approx (0.43375536, 0.11232478, -0.16477999)$.
\end{lemma}
\begin{proof}
    The unique fixed point can be checked by direct computation. Furthermore, one can directly compute $\norm{\mM}_2 < 1.032$ and $\norm{\mM}_\infty < 0.841$. Thus, by the Riesz-Thorin inequality, we have 
    \begin{align}
        \norm{Df(\vx)}_4^2 &\le \norm{Df(\vx)}_2 \norm{Df(\vx)}_\infty
        \le \underbrace{\norm{\mM}_2}_{< 1.032} \underbrace{\norm{\op{diag}(\op{sech}^2(\vx))}_2}_{\le 1} \underbrace{\norm{\mM}_\infty}_{< 0.841} \underbrace{\norm{\op{diag}(\op{sech}^2(\vx))}_\infty}_{\le 1} < 0.868
    \end{align}
    upon which taking square roots completes the proof of the lemma.
\end{proof}

To generate the figure, we ran the simplified algorithm for $\ell_p$-contractions described in \Cref{sec:lp-simplified}, for $p=4$. The algorithm, on every iteration $t$, generates a query point $\vx^{(t)} \in [-1, 1]^3$ by solving the VI associated with $\cG$; this is the point labeled with $t$ in the diagram. We say that a point $\vx^*$ is ``feasible'' at iteration $t$ if the gradient $G^{\vx^*} = \grad Q^{\vx^*}$ of the map $Q^{\vx^*}(\vx) := \norm{\vx - \vx^*}_4^4$ obeys all the linear cuts 
$$\ip{G(\vx^{(\tau)}), \vx^{(\tau)} - f(\vx^{(\tau)})} \le 0$$
for $\tau \le t$. The shaded region in the figure is the projection of the feasible region into the $xy$-plane. The feasible region was approximated by discretizing $[-1, 1]^3$ into a grid.

Notice that, while these cuts are linear in function space, they are nonlinear in $\vx^*$-space; hence, the feasible regions in the figure are in general nonconvex. Similarly, because the ellipsoid algorithm runs in function space and not $\vx^*$-space, the query point $\vx^{(t)}$ does not always lie in the shaded region, since not all feasible operators are actually of the form $\grad Q^{\vx^*}$. For example, the query point at Step 10 lies just outside the feasible set. 

\section{Omitted details from \Cref{sec:prelims}}\label{sec:proofs prelims}

\subsection{Proof of \Cref{th:eah}}
\begin{proof}
Run the ellipsoid algorithm (\eg, \citet{Grotschel93:Geometric}) on the following strong version of \eqref{eq:eah dual}.
\begin{align}
  \qq{find} \vy \in B_\ell(\cY, -\delta) \qq{such that} \ip{F(\vx), \vy} \ge \frac{\eps}{2} \qq{for all} \vx \in \cX. \tag{EAH-D}
\end{align}
where $\delta$ will be chosen later, and $B_\ell(\cY, -\delta)$ is the set of all points $\delta$-deep within $\cY$, that is, such that $B_\ell(\vy, \delta) \subseteq \cY$. Use initial ball $B_\ell(\vec 0, R)$. Notice that the semi-separation oracle provides a strong separation oracle for this dual (namely, run the semi-separation oracle with precision $\eps/4$). Therefore, after $\poly(\ell, \log(R, L, 1/\eps))$ iterations, the ellipsoid method concludes that \eqref{eq:eah dual} is infeasible, and outputs a certificate of this infeasibility. That is, ellipsoid outputs a distribution $\mu \in \Delta(\cX)$ such that 
\begin{align}
  \E_{\vx\sim\mu} \ip{F(\vx), \vy} \le \frac{\eps}{2} \qq{for all} \vy \in B_\ell(\cY, -\delta).
\end{align}
By standard convex ball arguments, since $B_\ell(\cY, -\delta)$ is nonempty for $\delta \le 1/R$, we have $\cY \subseteq B_\ell(B_\ell(\cY, -\delta), \tau)$ for $\tau = 2\delta R^2$. From this, and the bound $\norm{F(\vx)}_2 \le L$, it follows  that 
\begin{align}
  \E_{\vx\sim\mu} \ip{F(\vx), \vy} \le \frac{\eps}{2} + 2 \delta LR^2 \qq{for all} \vy \in \cY.
\end{align}
 Taking $\delta = \eps/(4LR^2)$ therefore yields the required bound. 

The support size bound follows from Carath\'eodory's theorem on convex hulls: there exists a distribution $\mu'$ of support at most $\ell+1$ such that $\E_{\vx\sim\mu'} F(\vx) = \E_{\vx\sim\mu} F(\vx)$, and such a distribution can be found efficiently by solving a linear program and taking an extreme-point solution.
\end{proof}
\subsection{Proof of \Cref{th:monotone}}

In this section, we give two different proofs of \Cref{th:monotone}.

\begin{proof}[First proof]
  \citet{Anagnostides26:Polynomial} give an algorithm that, in time $\poly(d, \log(R, L, 1/\eps))$, outputs either an $\eps$-approximate VI solution, or a distribution $\mu$ for which 
  \begin{align}
    \E_{\vx\sim\mu} \ip{G(\vx), \vx - \vy} < 0 \qq{for all} \vy \in \cX.
  \end{align}
In the former case we are done. In the latter case, taking $\bar\vx = \E_{\vx\sim\mu} \vx$, we have
    \begin{align}
        \E_{\vx\sim\mu} \ip{G(\vx) - G(\bar\vx), \vx - \bar\vx} < 0
    \end{align}
    so we can find a monotonicity violation by iterating over the support of $\mu$.
\end{proof}

\begin{proof}[Second proof]
  This more self-contained proof follows ideas from \citet{Nemirovski10:Accuracy}, but expressed in the language of this paper. First, consider the problem
  \begin{align}
    \qq{find} \mu \in \Delta(\cX) \qq{such that} \E_{\vx\sim\mu} \ip{G(\vx), \vx - \vy} \le \eps' \qq{for all} \vy \in \cX. \label{eq:evi}
  \end{align}
  where $\eps'$ will be chosen later.
  We will run EAH. The preconditions of \Cref{th:eah} are easily satisfied: $\cX$ is well-bounded, and semi-separation follows from the fact that a separation oracle on $\cX$ is given and setting $\vx = \vy$ is a good-enough response. Thus, by \Cref{th:eah}, there is a $\poly(d, \log(R, 1/\eps'))$-time algorithm that solves \eqref{eq:evi}. Let $\bar\vx$ be the expectation of $\mu$. Using the separation oracle, compute an $\eps'$-approximate solution $\vy \in \cX$ to the linear optimization problem
  \begin{align}
    \max_{\vy \in \cX} \ip{G(\bar\vx), \bar\vx - \vy}.
  \end{align} 
   Let $\vw = (1 - \eta) \bar\vx + \eta \vy$, for a step size $\eta$ to be chosen later. Let $\vx^{(1)}, \dots, \vx^{(T)}$ be the support of $\mu$. We claim that either $\bar\vx$ is an approximate VI solution, or one of the pairs $(\vw, \vx^{(t)})$ for $0 \le t \le T$ is a monotonicity violation, in which case we are immediately done because it is straightforward to search over all such pairs and output the violation or solution. We now prove the claim. Suppose that none of the pairs $(\vw, \vx^{(t)})$ is a monotonicity violation. Then we have
  \begin{align}
    \eta \cdot \ip{G(\vw), \bar\vx - \vy} = \ip{G(\vw), \bar\vx - \vw} = \E_{\vx\sim\mu} \ip{G(\vw), \vx - \vw} \le \E_{\vx\sim\mu} \ip{G(\vx), \vx - \vw} \le \eps'.
  \end{align}
  On the other hand, by $L$-Lipschitzness of $G$ and Cauchy-Schwarz, we have
  \begin{align}
    \ip{G(\bar\vx) - G(\vw), \bar\vx - \vy} \le L \cdot \norm{\bar\vx - \vw}_2 \cdot \norm{\bar\vx - \vy}_2 \le 4\eta LR^2.
  \end{align}
and therefore
\begin{align}
  \ip{G(\bar\vx), \bar\vx - \vy} \le 4\eta LR^2 + \frac{\eps'}{\eta} \le 4R\sqrt{L \eps'}
\end{align}
by choosing $\eta = \sqrt{\eps'/(4LR^2)}$. But since $\vy$ is an $\eps'$-approximate maximizer of $\ip{G(\bar\vx), \bar\vx - \vy}$, it follows that $\bar\vx$ is a $(4R\sqrt{L \eps'} + \eps' \le 8R\sqrt{L\eps'})$-approximate VI solution. Thus, taking $\eps' = \eps^2/(64R^2L)$ completes the proof.
\end{proof}

\section{Omitted details from \Cref{sec:main}}\label{sec:proofs main}

\subsection{Proof and generalization of \Cref{prop:m is nice}}

We start by proving \Cref{prop:m is nice}.
\begin{proof}
  The entries of $m(\vx)$ are bounded in absolute value by $R^{p}$, so $\norm{m(\vx)}_2 \le \sqrt{k} R^{p}$. For the Lipschitz constant, notice that each entry of $m(\vx)$ is a term $x_i^j$ for $i \in [d]$ and $j \in [p-1]$, which is $jR^{j-1}$-Lipschitz, so $m$ itself has Lipschitz constant at most $\sqrt{k} p R^p$. Finally, if $\hat\vx \in \cX$ then
  \begin{align}
    Q^{\hat\vx}(\vx) = \sum_{i=1}^d (x_i - \hat x_i)^p = \sum_{i=1}^d \sum_{j=0}^{p} \binom{p}{j} x_i^j (-\hat x_i)^{p-j}
  \end{align}
  is a polynomial in $\vx$ of degree $p$ and whose coefficients are all bounded by $ (2R)^p$, since $\binom{p}{j} \le 2^p$ and $|\hat x_i| \le R$. Thus, $\grad Q^{\hat\vx}(\vx)$ is monotone and has coefficients bounded by $p(2R)^p$, and therefore can be represented as $\mA m(\vx)$ for some $\mA$ with $\norm{\mA}_2 \le \sqrt{k}p (2R)^p$. 
\end{proof}

We now prove a version of \Cref{prop:m is nice} that accounts for other ``rotated'' norms, of the form $\norm{\vx} = \norm{\mP \vx}_p$.

\begin{proposition}\label{prop:m is nice with cross}
  Let $\cX \subseteq B_d(\vec 0, R)$, and $p, q$ be positive even integers with $q \le p$. Let  $m : \cX \to \R^k$ where $k = d^{O(p)}$ be given by all the moments of degree $p$ or smaller in the variables $x_i$. Then for $L = (2dR)^{O(p)}$, we have that $m$ is $L$-Lipschitz and $L$-bounded, and moreover for every $\hat\vx \in \cX$ and $\mP$ with $\norm{\mP}_2 \le 1$, we have $\grad Q^{\hat\vx, \mP} \in \cG(m,L)$, where $Q^{\hat\vx, \mP}(\vx) := \norm{\mP(\vx - \hat\vx)}_q^q$.
\end{proposition}
\begin{proof}
  The norm $\norm{m(\vx)}_2$ and Lipschitz constant of $m$ are bounded by $\sqrt{k} R^p$ and $\sqrt{k} p R^p$ respectively, by the same logic as the above proof. Moreover, we have
  \begin{align}
    Q^{\hat\vx, \mP}(\vx) = \sum_{i=1}^d \ab(\sum_{j=1}^d P_{ij}(x_j - \hat x_j))^q
  \end{align}
  The sum inside the parentheses has $2d$ terms each with a coefficient bounded by $R$ (since $|\hat x_j| \le R$); hence, after taking the $q$th power, the result is a sum of $(2d)^q$ terms each of which is bounded by $R^q$, so all coefficients are bounded by $(2dR)^q$. Its gradient therefore has all coefficients bounded by $p(2dR)^q$ and hence norm bounded by $\sqrt{k} p(2dR)^p \le (2dR)^{O(p)}$, as desired.
\end{proof}

\subsection{Existence of $\cG$-EFPs}\label{app:existence of EFPs}

    First, recall that $f$ maps $\cX \to \cX$. If $f$ is continuous, then any fixed point of $f$ is also a $\mathcal{G}$-EFP by definition, and hence $\cG$-EFPs always exist by Brouwer. For this paper, we only care about continuous $f$.
    
    If $f$ is discontinuous, we can smooth $f$ out: consider $$f_\delta(\vx):=\E_{\vy\sim D(\vx)}f(y)$$ where $D(\vx)$ is the uniform distribution on $B(\vx,\delta)\cap \cX$. Then $f_\delta$ is a continuous self-map on $\cX$, so by Brouwer it has a fixed point $\vx^*$, and therefore for any $1$-Lipschitz $1$-bounded $G$ we have
$$\begin{aligned}\E_{\vy\sim D(\vx^*)}\langle G(\vy),\vy-f(\vy)\rangle &= \E_{\vy\sim D(\vx^*)}\langle G(\vx^*),\vy-f(\vy)\rangle+O(\delta)\\
&= \ip*{ G(\vx^*),\E_{\vy\sim D(\vx^*)}[\vy-f(\vy)]}+O(\delta)\\
&= \ip*{ G(\vx^*),\E_{\vy\sim D(\vx^*)}[\vx^*-f(\vy)]}+O(\delta)\\
&=\langle G(\vx^*),\vx^*-f_\delta(\vx^*)\rangle+O(\delta)\\
&=O(\delta)\end{aligned}$$
where $O(\cdot)$ hides geometry-dependent constants. Therefore, $D(x^*)$ is an $O(\delta)$-approximate $\mathcal{G}$-EFP where $\mathcal{G}$ is the set of all $1$-Lipschitz $1$-bounded functions. This works for any $\delta > 0$, so $\epsilon$-approximate $\mathcal{G}$-EFPs exist for all $\epsilon > 0$.

Note that this does not imply that exact $\mathcal{G}$-EFPs always exist. Indeed, exact $\cG$-EFPs {\em do not} always exist, if $f$ is discontinuous: consider $f : [0, 1] \to [0, 1]$, $f(x) = \mathbbm{1}\{x = 0\}$ and $\cG = \{ x \mapsto x, x \mapsto -1 \}$.